%% file: main.tex
\documentclass[conference]{IEEEtran} 
\usepackage{amsmath,amssymb,amsfonts,amsthm}
\usepackage{array}
\usepackage{subcaption}
\usepackage{algorithm}
\usepackage{algpseudocode}
\usepackage{booktabs} 
\usepackage{cite}
\usepackage{graphicx}
\usepackage{mathtools}
\usepackage{multirow}
\usepackage[flushleft]{threeparttable}
\usepackage{textcomp}
\usepackage{xcolor}
\usepackage{xspace}

\def\A{\mathcal{A}}
\def\M{\mathcal{M}}
\def\S{\mathcal{S}}
\def\Nin{N_{\text{in}}}
\def\Nout{N_{\text{out}}}
\def\x{\mathrm{x}}
\def\eps{\varepsilon}

\DeclareMathOperator*{\argmax}{arg\,max}
\DeclareMathOperator*{\argmin}{arg\,min}

\DeclareMathOperator*{\In}{Inf}

\def\etal{\emph{et~al.}\xspace}

\newtheorem{theorem}{Theorem}
\newtheorem{definition}{Definition}
\newtheorem{observation}{Observation}

\newtheorem{corollary}{Corollary}

\newcommand{\denselist}{\itemsep 0pt\parsep=1pt\partopsep 0pt}
\newcommand{\bitem}{\begin{itemize}\denselist}
\newcommand{\eitem}{\end{itemize}}
\newcommand{\benum}{\begin{enumerate}\denselist}
\newcommand{\eenum}{\end{enumerate}}

\newcommand{\veryshortarrow}[1][3pt]{\mathrel{%
   \vcenter{\hbox{\rule[-.2pt]{#1}{.4pt}}}%
   \mkern-4mu\hbox{\usefont{U}{lasy}{m}{n}\symbol{41}}}}
\newcommand{\xz}[2]{#1 \,{\veryshortarrow[0.15 cm]} #2}   
\newcommand{\pxz}[2]{P\left(\xz{#1}{#2}\right)}

\pdfoutput=1

\newcommand{\lp}[2]{x_{#1}(#2)}

\def\BibTeX{{\rm B\kern-.05em{\sc i\kern-.025em b}\kern-.08em
    T\kern-.1667em\lower.7ex\hbox{E}\kern-.125emX}}

\begin{document}

\title{On Privacy of Socially Contagious Attributes}

\author{\IEEEauthorblockN{Aria Rezaei, Jie Gao}
\IEEEauthorblockA{
    Stony Brook University\\
    \{arezaei,jgao\}@cs.stonybrook.edu
}}




\maketitle


\begin{abstract}
\input{00_abstract.tex}
\end{abstract}

%
%



\input{10_introduction.tex}
\input{20_background.tex}

\input{30_problem.tex}
\input{40_method.tex}
\input{50_experiments.tex}
\input{60_conclusion.tex}

\section*{Acknowledgement}
Aria Rezaei and Jie Gao acknowledge support through NSF DMS-1737812, NSF CNS-1618391 and NSF CCF-1535900.

\bibliographystyle{IEEEtran}
\bibliography{main}

\end{document}

%% file: 00_abstract.tex
A common approach to protect users privacy in data collection is to perform random perturbations on user's sensitive data before collection in a way that aggregated statistics can still be inferred without endangering individual secrets.
In this paper, we take a closer look at the validity of Differential Privacy guarantees, when sensitive attributes are
subject to social contagion.
We first show that in the absence of any knowledge about the contagion network, an adversary that tries to predict the real values from perturbed ones, cannot train a classifier that achieves an area under the ROC curve (AUC) above $1-(1-\delta)/(1+e^\varepsilon)$, if the dataset is perturbed using an $(\varepsilon,\delta)$-differentially private mechanism. Then, we show that with the knowledge of the contagion network and model, one can do substantially better.
We demonstrate that our method passes the performance limit imposed by differential privacy.
Our experiments also reveal that nodes with high influence on others are at more risk of revealing their secrets than others.
Our method's superior performance is demonstrated through extensive experiments on synthetic and real-world networks.




%% file: 10_introduction.tex
\section{Introduction}
\label{sec:intro}


The last decade has witnessed the exponential growth of data collection practices.
While access to large-scale data has fueled the unprecedented power to solve problems previously thought impossible, it also imposes a great risk on the privacy of individuals in this new environment.
A common policy is to consider individual data items to be sensitive, while knowledge of aggregated statistics on a population is not. For example, the fact that a person has a certain disease is considered sensitive, while it is safe to release the percentage of people with that disease within a population.
This model has been the foundation of the popular differential privacy (DP) framework~\cite{Dwork2014-pq}, in which individual entries are sensitive but queries on aggregated knowledge are answered with the guarantee that an adversary cannot use the answers to accurately infer individual data items.

In this paper we examine the interplay of personally sensitive data in a social environment.
It has been widely recognized that social interactions shape the landscape of individual attributes -- infectious diseases spread through social interactions and contacts; behavior changes such as
obesity~\cite{Cohen-Cole2008-gi}, 
exercising~\cite{Hamari2015-ii},
or decision making processes such as 
voting~\cite{Bond2012-ym} or charity donation~\cite{Meer2011-rt} are contagious.

Due to the ubiquity of online social platforms in recent years, information about social ties and social interaction has become available. Such data can be available to the public with little effort (e.g.: professional affiliation on public web pages or
friendship networks in public social networks such as Twitter), or can be mined through other means such as human mobility traces~\cite{Wang2011-sy}. So the question we ask is: \emph{how safe are people's sensitive attributes in a socially connected world?}

\smallskip\noindent\textbf{Our Contribution.}
In this paper we answer this question by proposing a novel attack to users' sensitive attributes using information on social network connectivity, despite the fact that attributes are protected by DP mechanisms. 

Suppose that the individuals participate in a survey in which they are asked about the sensitive attribute $X$ with value $0$ or $1$. The goal of the survey is to learn the aggregated percentage of population who report ``1''. Since the participants may not trust the data collector, they use a randomized perturbation mechanism $\M$ to report data $z_i=\M(\x_i)$. A simple scheme for $\M$ is to flip a coin. If head, report $1$ or $0$ at random, otherwise report the true value. After aggregating the perturbed reports, one can approximate the true statistics by removing certain biases introduced by $\M$.
For example, if there are $p$ fraction in the population whose attributes are $1$, the perturbation mechanism leads to a total fraction of $1/4+p/2$ reporting $1$. From this, one can solve for $p$. Meanwhile, knowing $z_i$ is not enough to accurately determine $\x_i$ -- such protection can be formulated by differential privacy guarantee. 

Now assume that the attacker knows the social connections between individuals in the survey as well as the contagion model (how this attribute spreads through the social ties).
We propose an attack that exploits this information to infer the initial state $\{\x_i\}$, with a performance bound exceeding that which is guaranteed by DP. More accurately, we show \begin{itemize}
    \item For \emph{any} perturbation mechanism $\M$ that guarantees $(\eps, \delta)$-differential privacy, i.e., $\forall z, \x \neq \x'$,  
    \begin{equation}
    \Pr[\M(\x) = z] \leq e^\varepsilon \Pr[\M(\x') = z] + \delta,\end{equation}
    the best classifier from an attacker, \emph{without} information of the social ties and the contagion models, has the Area Under the ROC Curve (AUC) at most $1 - (1-\delta)/(e^\varepsilon + 1)$. 
    \item Next, we propose a method to infer the original sensitive values $\{\x_i\}$, using the contagion model, the network structure, the perturbation mechanism $\M$ and the noisy reported values $\{z_i\}$. This requires understanding how the real values correlate by accurately modeling the way they are produced by a contagion process.
    In prior work, contagions are ignored in modeling correlation between individual values, which results in models that are too simplistic to reflect real-world phenomena. In contrast, our model incorporates the network structure and contagion model directly into our calculations.

    We proceed in two phases. First, we find the probabilistic \emph{effectors} -- a probability $\alpha_v$ that each node $v$ is an initiator of a contagion that results in
    observed $\{z_i\}$.
    Next we run the contagion model forward from the probabilistic seeds to estimate $\{\x_i\}$. 
    
\item Our experiments on both synthetic and real-world networks show that our method can achieve an AUC value higher than the limit imposed by DP, weakening its guarantee as a result. This also means that the social network information, while not sensitive itself, can indeed be exploited to infer sensitive knowledge. We also observe that nodes with high influence over others are more vulnerable to such attacks.
\end{itemize}

In what follows, we first present the background and related work on DP and its many variants along with prior work on social influence and contagions. We then report the theoretical upper bound on the performance of a binary classification with differential privacy protection, when there is no knowledge of the contagion network. This is followed by our attack leveraging social network structure and contagion model. We report the results of our experiments at the end.

%% file: 20_background.tex
\section{Background and Related Work}\label{sec:background}


\smallskip\noindent\textbf{Contagion Models}:
Many attributes are socially contagious. The way that these attributes spread in a social network is described via contagion models. A few models have received great attention in the literature. 
In the \emph{Linear Threshold} model each edge has a weight that represents the influence between nodes, and nodes are activated when the sum of influence receiving from their neighbors exceeds a threshold randomly selected from $[0,1]$.
\emph{Independent Cascade} model assumes that each node $u$, upon acivation, has \emph{one} chance to activate each of its neighbors, with different probabilities.
This differs from prior work on virus contagion such as SI (susceptible-Infected)
where activated node \emph{continuously} try to activate their inactive neighbors in time-synchronous rounds. 
Recently there has been growing attention in \emph{General Threshold} model~\cite{Gao2016-gy} (first proposed by Granovetter~\cite{Granovetter1973-qm}) and \emph{Complex Contagions}, where infection requires a specific number of infected neighbors~\cite{Centola2007-vj,Ghasemiesfeh2013-bl}.
In this work we mainly use the Linear Threshold model and discuss possible extensions in the last section. 

\smallskip\noindent\textbf{Data Privacy}:
The most widely adopted privacy model is the model of differential privacy (DP)~\cite{Dwork2014-pq}, which imposes constraints on publishing aggregate information about a database such that the privacy impact on individual entries is limited. Specifically, a randomized algorithm $\A$ that takes a dataset as input is said to have $(\eps, \delta)$-differential privacy, if for all datasets $D_1$ and $D_2$ that differ on a single entry, and all subsets $S$ of the image of $\A$: $\Pr[\A(D_1) \in S] \leq e^\varepsilon \Pr[\A(D_2) \in S] + \delta$.
The probability is taken over the randomness of the algorithm. 


The original DP model does not explicitly specify the ramifications of the presence of correlation between data points, which could be accessible from outside. In fact, it is proved that when the data is assumed to be correlated, the privacy guarantees provided by DP becomes weaker~\cite{Kifer2014-bj}.
This issue is acknowledged in a number of later definitions that try to address it. For example, inferential privacy~\cite{Ghosh2016-gs} captures the largest possible ratio between the posterior and prior beliefs about an individual's data, after observing the results of a computation on a database. Here the data items may not be independent and the correlation is captured by a prior belief on the data items. 
In adversarial privacy~\cite{Rastogi2009-lx}, domain experts could plug in various data generating distributions and the goal is to protect the presence/absence of a tuple in the data set. 
The most general definition is PufferFish privacy~\cite{Kifer2014-bj}, where one explicitly specifies the set of secrets to protect, how they shall be protected (by specifying indistinguishable pairs), how data evolve or are generated (e.g., are data items correlated), and what extra knowledge the potential attackers have.
These definitions add to the complexity of DP and, as a result, we have yet to see any of them being as widely adopted as DP.

Our work could be considered as a motivation to devise mechanisms that explicitly incorporate contagion models into their privacy protection guarantees.
In an age when activities are shared online and individuals interact with each other more as each day passes, it is not unimaginable that an adversary might have access to social information that can jeopardize individuals' secrets.
We propose a concrete attack that beats DP guarantees. Since many attributes shared by humans are socially correlated, this work reopens many of the problems studied under traditional DP setting (with no explicit assumption on data correlation and generation processes) and forces us to inspect new methods that can endure higher levels of scrutiny.


\smallskip\noindent\textbf{Analyzing Social Contagions and Finding Effectors}: 
Our attack is closely related to analyzing social contagions, in particular the two problems of influence maximization and finding roots of contagion. 

Influence maximization is initially studied in~\cite{kempe2003maximizing}: how to pick $k$ initial seeds such that the number of nodes eventually infected is maximized. It is an NP-hard problem but could be approximated up to $1-1/e$, if one can have an oracle for computing the \emph{influence} of a set of seeds $S$ -- the (expected) number of nodes infected with seed set $S$. Obviously one can run simulations to estimate the influence of a seed set. Computing the exact influence of a node can be done in linear time on a DAG but is $\#$P-hard on a general graph~\cite{Chen2010-lt}.
A heuristic to speed-up the algorithm is to utilize local simple structures, such as local DAGs, to estimate the influence of a node~\cite{Chen2010-lt}. 
Alternatively, Borgs \etal~\cite{Borgs2013-dz} proposed to use the reverse cascades to estimate influence, picking nodes that more frequently appear in cascades simulated in reverse direction.

Given the current activation state of a contagion in a social network, the \emph{$k$-effector problem} is to find the most likely $k$ effectors (initiators of contagions)
You can see that influence maximization is a special case of this problem where all nodes are activated in the end.
The $k$-effector problem is NP-hard for general graphs or even a DAG, but is solvable in polynomial time by dynamic programming on trees~\cite{lappas2010kEffectors}. For general graphs, a heuristic algorithm~\cite{Luo2013-ku} is to extract the most probable tree (which is NP-hard) and run the optimal algorithm on that tree. 
Finding effectors is also extensively studied for the Susceptible-Infected (SI) propagation model~\cite{Nguyen2016-dt, Prakash2012-eu}.

\smallskip\noindent\textbf{Privacy of Social Networks and Attributes}:
Our work is different from previous work on protecting social network privacy, which assumes that the social network graph itself is private data and network-wide statistics (e.g., degree distribution) is released~\cite{Task2012-hi}. We assume that the social network structure is publicly available and only the socially contagious attributes are sensitive.

Links between individuals in a social network can be telling. For instance,
Kifer and Machanavajjhala~\cite{Kifer2011-pu} show that future social links can be predicted from the number of inter-community edges by assuming that network evolution follows some particular model. Somewhat similar to our work, Song \etal~\cite{Song2017-xf} considered flu infection -- estimating how many people get flu while preventing the status of any particular individual being revealed. To avoid the intricate details of social contagion, they assume an overly simplistic model where all nodes in the same connected component are correlated and in each component, all pairs of nodes are equally correlated.
In our work, we assume a contagion model that is aligned with established literature on contagion and social influence. 



%% file: 30_problem.tex
\def\M{\mathcal{M}}
\def\S{\mathcal{S}}
\def\Nin{N_{\text{in}}}
\def\Nout{N_{\text{out}}}
\def\x{\mathrm{x}}
\def\X{X}
\def\Z{Z}

\section{Problem Definition}\label{sec:problem}

For a population of $n$ individuals, let $\x_i$ be a sensitive binary attribute for individual $i$, $\x_i \in \{0,1\}$, and denote by $\X$ the set of all values $\langle \x_1, \cdots, \x_n\rangle$.
We assume that this attribute is contagious and propagates over a directed network $G(V, E)$ following the Linear Threshold cascade model.
In this model, each edge has a weight $w(u,v) \in (0, 1]$ which represents the influence that node $u$ exerts on node $v$. Each node also has a threshold $\lambda_v$ which is selected \emph{uniformly at random} from $(0, 1]$. 
If the sum of influence from infected in-going neighbors goes beyond $\lambda_v$, $v$ becomes activated in the next round. 
Assuming that the set of activated nodes, $A$, is not empty at time $0$, we can build it iteratively at every step via the following rule:
\begin{equation}\label{eq:linear_threshold}
    A \leftarrow A \cup \big\{ v \in V \setminus A \colon \sum_{u \in \Nin(v) \cap A} w(u,v) \geq \lambda_v \big\}.
\end{equation}
Here $\Nin(v)$ is the set of neighbors with edges pointing to $v$ (i.e., imposing influence on $v$).
The process proceeds until $A$ stops growing.

Imagine that these individuals participate in a survey in which they are each asked about their sensitive attribute $\x_i$. The goal of this survey is to calculate some aggregate statistic, e.g., the percentage of individuals having attribute $1$.
To avoid revealing their secrets, they could
use a randomized perturbation mechanism, $\M\colon \{0,1\} \rightarrow \{0,1\}$, and use the resulting values to answer the survey. The observed answer of participants is the sequence $\langle z_1, \cdots, z_n \rangle$, denoted by $\Z$, where $z_i = \M(\x_i)$. Assume that $\M$ guarantees $(\varepsilon,\delta)$-DP, i.e.,
\begin{equation}\label{eq:dpbound}
    \forall z, \x \neq \x' \colon \Pr[\M(\x) = z] \leq e^\varepsilon \Pr[\M(\x') = z] + \delta.
\end{equation}

In this paper we want to examine two problems:
\bitem
    \item What is the performance of the best classifier, using only information in $\Z$ and $\M$, to infer the true values $\X$?
    \item If we also know the contagion network $G$ and the contagion model, can we
    perform
    better? In other words, how much more information is revealed by knowing the social structure and the way this sensitive attribute propagates in the network?
    The difference from the answer to the earlier question is the loss of privacy.
\eitem

\section{Limitations of Binary Classification with Differential Privacy}\label{sec:limitation}
To show that the presence of the underlying contagion network provides
essential information that can pose a real threat to privacy, we study the limits of binary classification given \emph{only} the
reported values ($\Z$) and the randomization parameters of $\M$. This is a fair assumption since the real values, $\X$, are never disclosed but $\Z$ is, and $\M$ is known to all participants.

A classifier scores and subsequently ranks the participants based on their likelihood of having $\x = 1$. We measure the success of such ranking by the probability that a randomly selected sample with $\x = 1$ (a positive sample) is ranked higher than a randomly selected sample with $\x=0$ (a negative sample).
This is known to be the area under the receiver operating characteristic curve (ROC curve) in an unsupervised classification problem, namely the AUC value~\cite{roc_curve}.
\begin{theorem}\label{theo:roc}
Any classification attempt by an adversary, having access to only $\Z$ and $\M$, will have an Area Under the ROC Curve (AUC) at most $1 - (1-\delta)/(e^\varepsilon + 1)$.
\end{theorem}

\def\prank{P(R_1 > R_0)}

Recall that the ROC curve of a classifier is plotting the true positive rate (TPR) against the false positive rate (FPR) at various threshold settings. AUC can be understood as the probability that the classifier ranks $R_1$ higher than $R_0$, denoted by 
$\prank$, where $R_1$ ($R_0$) is a randomly chosen positive (negative) sample, with $\x=1$ ($\x=0$). 

Suppose we take a positive (negative) sample $\x$ ($\x'$) and the perturbation mechanism $\M$ produces a perturbed value $z$ ($z'$). Let's denote $\Pr(\M(\x) = z)$ by $\pxz{\x}{z}$ for brevity. Let $\S_1$ and $\S_0$ be two distributions over $(-\infty, +\infty)$ from which a score is drawn, 
if $z = 1$ or $z = 0$ respectively.
For the perturbed value $z, z'$, the classifier chooses a score $s, s'$ from $S_z, S_{z'}$ respectively and the ranking is produced based on the scores. Denote by $\gamma(z,z')$ the probability that $s$ is higher than $s'$, i.e., $\Pr[s>s' | s\sim \S_z, s' \sim \S_{z'}]$. Obviously,
\begin{equation*}
\Pr[s = s' | s\sim \S_z, s' \sim \S_{z'}] = 1 - \gamma(z,z') - \gamma(z',z). 
\end{equation*}
Then we can write $\prank$ as (Section 2 of ~\cite{roc-dinasour}):
\begin{equation*}
\sum_{z, z' \in \{0,1\}}
\pxz{1}{z}\pxz{0}{z'}
\big(P(s > s') + \frac{1}{2} P(s = s')\big).
\end{equation*}
Continuing the above, we have:
\begin{align}
\prank \nonumber\\
    ={}&\sum_{z,z'} \pxz{1}{z}\pxz{0}{z'}
    \left( \frac{1 + \gamma(z,z') - \gamma(z',z)}{2}\right) \nonumber\\
    ={}& \frac{1}{2}\sum_{z,z'} \pxz{1}{z}\pxz{0}{z'}
    \big(\gamma(z,z') - \gamma(z',z)\big) + \frac{1}{2} \nonumber\\
    ={}& \Big(\pxz{1}{1}\pxz{0}{0} - \pxz{0}{1}\pxz{1}{0}\Big)\cdot \nonumber \\
    {}& \qquad \qquad  \frac{\gamma(1,0) - \gamma(0,1)}{2} + \frac{1}{2}
\label{eq:p_m}
\end{align}


\begin{observation}\label{lem:opt_gamma}
Let $\gamma^*$ be the one maximizing AUC, i.e.,   $\argmax_{\gamma} \prank$. Then, $\gamma^*(1,0) = 1$, if
\begin{equation}\label{eq:condition}
    \pxz{1}{1}\pxz{0}{0} > \pxz{1}{0}\pxz{0}{1}
\end{equation}
and $0$ otherwise. 
\end{observation}
The above is clear from the right hand side of \eqref{eq:p_m}.
This shows that an optimal AUC is achieved by a \emph{deterministic} classification rule based solely on the condition in Observation~\ref{lem:opt_gamma}. 
\begin{corollary}\label{cor:bayesian}
Bayesian inference achieves optimal AUC.
\end{corollary}
\begin{proof}
By Bayes' rule we have:
\begin{equation*}
    \pxz{0}{z} = \frac{P(z) - \pxz{1}{z}P(\x = 1)}{P(\x = 0)}.
\end{equation*}
Using the above, we can substitute $\pxz{0}{1}$ and $\pxz{0}{0}$ in \eqref{eq:condition}. After canceling out phrases from both sides, we have:
\begin{align*}
    \frac{\pxz{1}{1}P(z=0)}{P(\x=0)} >{}& \frac{\pxz{1}{0}P(z=1)}{P(\x=0)}\\
    \frac{P(\x=1,z=1)}{P(\x=1)P(z=1)} >{}&
    \frac{P(\x=1,z=0)}{P(\x=1)P(z=0)} \\
    P(\x=1 \mid z=1) >{}& P(\x=1 | z = 0).
\end{align*}
The proof is symmetrical for the reverse inequality.
\end{proof}
We can now prove Theorem~\ref{theo:roc}. Without loss of generality, we assume that the condition in Lemma~\ref{lem:opt_gamma} holds, we can then further simplify \eqref{eq:p_m} as below:
\begin{align}
    \Pr(R_1 > R_0) &= 
    \frac{1}{2} + \frac{1}{2}\Big(\pxz{1}{1} + \pxz{0}{0} - 1\Big) \nonumber\\
    &= \frac{\pxz{1}{1} + \pxz{0}{0}}{2}
\end{align}
By the $\varepsilon$-DP guarantees we have:
\begin{gather}\label{eq:dp_bounds}
    \pxz{1}{1} + e^\varepsilon \pxz{0}{0} \leq e^\varepsilon + \delta \nonumber\\
    \pxz{0}{0} + e^\varepsilon \pxz{1}{1} \leq e^\varepsilon + \delta,
\end{gather}
As a result, we have:
\begin{equation}
    \Pr(R_1 > R_0) \leq 1 - \frac{1 - \delta}{1 + e^\varepsilon}.
\end{equation}
Thus, Theorem~\ref{theo:roc} is proved. Note that the bound is realized if the inequalities in Equation~\eqref{eq:dpbound} become equality. This theorem shows that \emph{if} this bound is significantly surpassed, the guarantee of $\varepsilon$-DP no longer holds.


%% file: 40_method.tex
\section{Algorithm}\label{sec:method}

\subsection{Objective Function}
The goal is to infer $P(\x_v = 1)$ for all $v$. We denote this probability by $x_v$ throughout this paper (note the difference between $\x$ and $x$). To do this, we first find the initial seeds of contagion, then calculate the corresponding $x_v$. Our solution is hence an arrangement of probabilities of each node $v$ being initially active, denoted by $\alpha_v$.
Rather than a fixed number of most likely seeds,
we seek to find a \emph{distribution} of initial seeds that are likely to produce the observed $\Z$.
This is shown to significantly boost our performance.
We now define the main objective for our problem.
\begin{definition}[Symmetric Difference]
Given two instances of reports, $Z_1$ and $Z_2$, we define their \emph{Symmetric Difference} by:
\begin{equation}\label{eq:sym_dif}
    D(Z_1, Z_2) = |Z_1\setminus Z_2| + |Z_2 \setminus Z_1|.
\end{equation}
\end{definition}
Let $\alpha$ be an assignment of the initial activation probabilities. Suppose that $\mathcal{C}$ is the distribution of all possible cascades, $C$, and $\mathcal{R}$ is the distribution of all possible reports, $\Tilde{Z} = (\Tilde{z}_1, \cdots, \Tilde{z}_n)$. Then, the \emph{expected} symmetric difference between $\Tilde{Z}$ and the originally observed values, $Z$, will be as below:
\begin{align}
&     \mathbb{E}\left[D(\Tilde{Z}, Z)\right]\nonumber \\
    ={}&
        \sum_{\Tilde{Z} \sim \mathcal{R}}\Pr(\Tilde{Z})D(\hat{Z},Z)\nonumber\\
        ={}& \sum_{C \sim \mathcal{C}} \Pr(C)\sum_{v \in V}\Pr(\Tilde{z}_v \neq z_v \mid C)\nonumber\\
        ={}& \sum_{C \sim \mathcal{C}} \sum_{v \in V}\sum_{\x \in \{0,1\}} \Pr(\Tilde{z}_v \neq z_v \mid \x_v = \x)\Pr(\x_v = \x \mid C)
       \nonumber\\
        ={}& \sum_{v\in V}\sum_{\x \in \{0, 1\}}\Pr(\Tilde{z}_v \neq z_v \mid \x_v = \x)\Pr(\x_v = \x)\nonumber\\
        ={}& \sum_{v}\big(\pxz{1}{{\sim} z_v}
        - \pxz{0}{{\sim}z_v}
        \big)x_v + \pxz{0}{{\sim}z_v}
        \label{eq:exp_sym_dif}
\end{align}
In the above ${\sim}z_v = 1 - z_v$, and the last line is due to $\Pr(\x_v = 0) = 1 - x_v$.
We define our objective function as $f = \mathbb{E}\big[D(\hat{Z}, Z)\big]$ and find an $\alpha$ that minimizes $f$. Since $\pxz{0}{{\sim}z_v}$ is a constant, we can further simplify $f$ as
\begin{equation}\label{eq:f}
f = \sum_{v \in V}c_v x_v \quad \text{s.t.}\, 0 \leq x_v \leq 1,
\end{equation}
where $c_v=\pxz{1}{{\sim} z_v} - \pxz{0}{{\sim}z_v}$.
\subsection{Bounds on $X$}\label{sec:bound}
\def\tp{\Tilde{P}}
Suppose that $X^* = (\x^*_1, \cdots, \x^*_n)$ are the real attribute values. 
\begin{theorem}
Let $\tp(z)$ be the fraction of vertices reporting $1$ and $c = \pxz{1}{1} - \pxz{0}{1}$, 
\[
\tp(x) = \frac{\tp(z) - \pxz{0}{1}}{c}.
\]
Then, with high probability\footnote{If $\lim_{n \rightarrow \infty}P(a) = 1$, $a$ happens with high probability.}:
\begin{equation}
    \left|\tp(x) - \frac{1}{n}\sum \x^*_v\right| \leq \sqrt{\frac{\log{n}}{2nc^2}}.
\end{equation}
\end{theorem}
\begin{proof}
We can treat $\tp(z)$ as the mean of $n$ random variables, representing individual acts of reporting $0$ or $1$. The \emph{expected} value of $\tp(z)$ can be written as:
\begin{equation}\label{eq:exp_p_z}
    \mathbb{E}\left[\tp(z)\right] = \frac{1}{n}\sum \x^*_v \pxz{1}{1} + (1-\x^*_v)\pxz{0}{1}.
\end{equation}
Since $0 \leq z_v \leq 1$ and each individual report is independent of others, we can apply Chernoff's 
bound. By using \eqref{eq:exp_p_z} we have:
\begin{multline}\label{eq:p_z_bound}
   \Pr\Bigg[\Bigg| \frac{1}{n}\sum \Big(\x^*_v \pxz{1}{1} + \\(1-\x^*_v)\pxz{0}{1}\Big ) - \tp(z)\Bigg| \geq \epsilon \Bigg] \leq e^{-2n \epsilon^2},
\end{multline}
Using the definition of $\tp(x)$ in \eqref{eq:p_z_bound}, we have:
\begin{equation}\label{eq:p_x_bound}
    \Pr\left[\left|\tp(x) - \frac{1}{n}\sum \x^*_v\right| \geq \epsilon'\right] \leq e^{-2n {\epsilon'}^2 c^2}.
\end{equation}
Recall that $c = \pxz{1}{1} - \pxz{0}{1}$. The probability above is asymptotically zero when:
\begin{equation}
    \epsilon' = \sqrt{\frac{\log{n}}{2nc^2}}
\end{equation}
\end{proof}
The value of $\tp(z)$ can be estimated from data by $\left|\left\{v \in V\colon z_v = 1\right\}\right|/n$. We can now update our objective function
to accommodate this new constraint:
\begin{equation}\label{eq:obj_constraint}
\begin{aligned}
    \hat{\alpha} &{}= \argmin_{\alpha} \; \; \sum_{v \in V} c_v x_v\\
    \text{s.t.},\, & 0 \leq x_v \leq 1,\,  \left|\frac{1}{n}\sum_{v \in V} x_v - \tp(x)\right| \leq \sqrt{\frac{\log{n}}{2nc^2}},
\end{aligned}
\end{equation}
Although our solution finds soft probabilities ($x_v \in [0, 1]$) instead of discrete values ($\x_v \in \{0,1\}$), our experiments show that having this constraint can increase the accuracy of inferred values, especially when the amount of added noise is not extremely high (DP's $\varepsilon$ is not extremely low).

\subsection{Modelling Contagion}\label{sec:contagion_model}

With $\alpha$, we want to derive a formula for $x_v$, the probability that node $v$ is active in the end.
Computing the influence of contagion given a fixed $\alpha$
can be done in linear time for a DAG, using the following formula: 
(Lemma 3 \cite{Chen2010-lt}).
\begin{equation}\label{eq:ap_dag}
    x_v = \alpha_v + (1 - \alpha_v)\sum_{u \in \Nin(v)} w(u,v)x_u.
\end{equation}
Since the original graph $G$ is not necessarily a DAG, we find \emph{local DAGs} containing nodes who impose high influence. 
In this way, we try to benefit from the structural simplicity of DAGs, while losing minimal information.
The approach of using local structures to approximate the influence in a general graph has been widely used in prior works in the context of influence maximization~\cite{Chen2010-lt,local-arbor,local-neighbor}.
\def\nmax{N_{\text{max}}}
\begin{algorithm}[!tbh]
 \caption{Local DAG with target $t$. (Algorithm 3 \cite{Chen2010-lt})}
 \label{algo:ldag}
\begin{algorithmic}[1]
    \Require $G(V,E)$, Node $t$, $\eta$: Threshold for $\In$, $\nmax$: Max allowed nodes.
	\Ensure $D_t(V_t, E_t)$: The DAG around node $v$.
    \State \textbf{Initialization: } $V_t=\emptyset,\; E_t = \emptyset, \; \forall v \in V\colon \In(v,t) = 0,\; \In(t,t) = 1.$
    \While{$\max_{v\in V \setminus V_t}\In(v,t) \geq \eta \text{ and } |V_t| \leq \nmax$}
        \State $u \gets \argmax_{v\in V \setminus V_t}\In(v,t)$
        \State $E_t \gets E_t \cup \{(u,v)\mid v \in \Nout(u) \cap V_t\}$
        \State $V_t \gets V_t \cup u$
        \For{$v \in \Nin(u)$}\Comment{Neighbors' $\In$ is updated.}
            \State $\In(v, t) \; {+}{=} \; w(v, u) \In(u, t)$
        \EndFor
    \EndWhile
    \State \Return $D_t(V_t, E_t)$
\end{algorithmic}
\end{algorithm}

Algorithm \ref{algo:ldag} starts by the DAG $D_t$ containing only $t$. We then calculate the influence of each node $v$ on $t$, which
is the activation probability of $t$ if only $v$ was initially active and influence would only spread through nodes already in $D_t$. This is denoted by $\In(v,t)$.
At each step, the node outside of $D_t$ with highest $\In(.)$ is added to $D_t$ along with its outgoing edges that connect to nodes already in $D_t$. This is to ensure that the final $D_t$ is a DAG. We then update the influence of incoming neighbors of $v$ that are not yet in $D_t$. There can be two stopping criteria to the growing process: (1) When the influence of the most influential node falls below a threshold $\eta$, or (2) the number of nodes in $D_t$ grows bigger than a maximum allowed number, $\nmax$. If implemented using an efficient priority queue for $\In(v,t)$ values, Algorithm \ref{algo:ldag} runs in $O(|E_t|\log{|E_t|})$ time.
Among nodes in $D_t$, the \emph{local} activation probability, $\lp{t}{v}$, is as below:
\begin{equation}\label{eq:local_prob}
\lp{t}{v} = \alpha_v + (1 - \alpha_v)\sum_{u \in \Nin(v) \cap V_t} w(u,v)\lp{t}{u}.
\end{equation}
Note that $\lp{t}{t}$ is the probability that $t$ is activated only through nodes that are most influential on it and, as a result, can be considered to be a reasonable approximations of $x_t$. Our experiments show that this approach in selecting DAGs is essential to achieving high-quality results, and superior to alternative approaches.

Now we can move on to optimizing the objective function $f$ in Equation~\eqref{eq:f}. More specifically, we need to find $\frac{\partial f}{\partial \alpha_v}$ for all $v$.
Chen et al. have established that in a DAG, there is a linear relationship between $\lp{t}{t}$ and $\lp{t}{v}$ for all $v \in V_t$~\cite{Chen2010-lt}. The linear factor, which is equal to $\frac{\partial \lp{t}{t}}{\partial \lp{t}{v}}$ is computed as below:
\begin{equation}\label{eq:partial_x}
\frac{\partial \lp{t}{t}}{\partial \lp{t}{v}} = \sum_{u \in \Nout(v) \cap V_t} w(v,u)(1 - \alpha_u)\frac{\partial \lp{t}{t}}{\partial \lp{t}{u}}.
\end{equation}
We can compute the above for all nodes $v \in D_t$ by initially setting $\frac{\partial \lp{t}{t}}{\partial \lp{t}{t}}$ as $1$ and then going through nodes in reverse topological order. Finding this ordering and computing the partial gradients each takes $O(|E_t|)$ time. Next, we find gradients of $f$ based on each $\alpha_v$. Let $I_v = \{t \in V\colon v \in V_t\}$.
Then, by taking the gradient of \eqref{eq:f} and applying the chain rule we can write:
\begin{align}\label{eq:partial_f}
\frac{\partial f}{\partial \alpha_v} &= \sum_{t \in I_v} c_t \frac{\partial \lp{t}{t}}{\partial \alpha_v}\nonumber= \sum_{t \in I_v} c_t \frac{\partial \lp{t}{t}}{\partial \lp{t}{v}}\frac{\partial \lp{t}{v}}{\partial \alpha_v}\nonumber\\
&= \sum_{t \in I_v} c_t \frac{\partial \lp{t}{t}}{\partial \lp{t}{v}} \left(1 - \sum_{u \in \Nin(v) \cap V_t}\lp{t}{u}w(u,v)\right).
\end{align}
The last line is produced by taking a derivative of \eqref{eq:local_prob} by $\alpha_v$. Since $\alpha_v$ does not have any effect on the activation probability of predecessors of $v$ in any DAG, we can treat the summation on the right side of \eqref{eq:local_prob} as a constant with respect to $\alpha_v$. Calculating this summation is possible by dynamic programming when nodes are visited in their topological ordering. Computing values in \eqref{eq:partial_f} and \eqref{eq:partial_x} for all DAGs has a collective runtime of $O(\sum |E_t|)$.

%% file: 50_experiments.tex
\section{Experiments}\label{sec:experiments}

In this section, we test our method on both synthetic and real-world networks. We demonstrate that our proposed constraints in Section~\ref{sec:bound} and greedily retrieved DAGs described in Section~\ref{sec:contagion_model} play a key role in maintaining a high quality for our results.
We also investigate attributes that can indicate how vulnerable nodes are to such attacks, namely in-degree, out-degree and PageRank.

\def\codag{\texttt{CO-DAG}}
\def\odag{\texttt{O-DAG}}
\def\cornd{\texttt{CO-RND}}
\def\ornd{\texttt{O-RND}}
\def\lappas{\texttt{Lappas+}}
\def\bayes{\texttt{Bayesian}}

\subsection{Methods}\label{sec:baseline}
We tested the following methods in our experiments:
\benum
\item \textbf{\codag}: Our main method, which optimizes our objective subject to the constraints in \eqref{eq:obj_constraint} using DAGs retrieved by Algorithm~\ref{algo:ldag}.
\item \textbf{\odag}: Similar to \texttt{CO-DAG}, but without enforcing the constraint on $\sum x_v/n$.
\item \textbf{\cornd}: To show that our selected DAGs are essential to the high quality of our results, we repeat the experiments with a method similar to \codag~ but with DAGs that grow by adding random neighbors of nodes already in the DAG, until the number of nodes reaches a threshold $\nmax$.
\item \textbf{\ornd}: Similar to \cornd, but without enforcing the constraint on $\sum x_v/n$.
\item \textbf{\lappas}~\cite{lappas2010kEffectors}: The algorithm to find $k$-effectors, when $k$ is known beforehand. Note that by knowing $k$, this method has access to more information compared to others.
\item \textbf{\bayes}: Simple Bayesian inference as described in Corollary~\ref{cor:bayesian}.
\eenum

\subsection{Datasets}\label{sec:dataset}
We now introduce the datasets used in our experiments.

\smallskip\noindent\textbf{Synthetic Networks}:
We use $4$ types of randomly generated networks:
\benum
\item \textbf{Core-Periphery}~\cite{core-periphery}: A random network where nodes in a \emph{periphery} are loosely connected to a dense center. These networks are generated as Kronecker graphs~\cite{kronecker} with matrix parameter $\left[0.9, 0.5; 0.5, 0.3\right]$.
\item \textbf{Erdos-Renyi}: A random network in which all possible edges have equal probability $p$ to appear. We set $p$ such that the expected out-degree of every node will be $5$.
\item \textbf{Power-law}: A network with a power-law degree distribution where $f(x) \propto x^{-\gamma}$. We set $\gamma = 1$ to produce a degree sequence and ran configuration model~\cite{configuration-model} to obtain a network. 
\item \textbf{Hierarchical}~\cite{hierarchical}: Random hierarchies generated as Kronecker graphs with matrix parameter set to $\left[0.9, 0.1; 0.1, 0.9\right]$.
\eenum

We generate networks with $500$ nodes, remove from it self-loops and nodes having both in and out-degrees less than $3$ (except for Hierarchical networks).
We assign random influence weights in the according to Section V.A of \cite{Chen2010-lt}: Random numbers between (0, 1] are assigned to edges, then incoming links to each node is normalized to sum to 1.

\begin{table}[!htb]
    \caption{Real-world networks.}
    \centering
    \begin{tabular}{lcc}
        \toprule
         Network & \#Nodes & \#Edges \\
         \midrule
         GrQc & $2{,}422$ & $21{,}842$ \\
         HepTh & $4{,}909$ & $38{,}704$ \\
         Amazon Videos & $1{,}598$ & $8{,}402$ \\
         Amazon DVDs & $9{,}488$ & $55{,}146$ \\
         \bottomrule
    \end{tabular}
    \label{tab:stats}
\end{table}

\begin{table*}[!tbh]
\centering
\caption{AUC values of $7$ methods for $3$ values of $\beta$ across all networks. Note that (1) our method outperforms others in almost all cases and (2) \lappas~ fails to pass the upper bound in almost all cases, despite having access to more information.}
\begin{threeparttable}
\begin{tabular}{c@{\hskip 0.3in}lccccccccc}
\toprule
& Network & $\beta$ & $\varepsilon$ & Upper Bound & \bayes & \codag & \odag & \cornd & \ornd & \lappas \\
\midrule
\multirow[c]{20}{*}{\rotatebox[origin=c]{90}{\Large{Synthetic}}}
& \multirow[c]{5}{*}{Core-Periphery} & $0.1$ & $0.201$ & $0.550$ & $0.550$ & $0.575$ & $0.597$ & $0.568$ & $0.587$ & $\mathbf{0.614}$\\ 
& & $0.3$ & $0.619$ & $0.650$ & $0.648$ & $\mathbf{0.716}$ & $0.708$ &  $0.689$ & $0.699$ & $0.657$\\ 
& & $0.5$ & $1.099$ & $0.750$ & $0.751$ & $\mathbf{0.833}$ & $0.803$ & $0.802$ & $0.793$ & $0.695$\\ 
& & $0.7$ & $1.735$ & $0.850$ & $0.847$ & $\mathbf{0.904}$ & $0.870$ & $0.887$ & $0.868$ & $0.706$\\ 
& & $0.9$ & $2.944$ & $0.950$ & $0.951$ & $\mathbf{0.967}$ & $0.955$ & $0.965$ & $0.957$ & $0.711$\\ 
\cmidrule{2-11}
& \multirow[c]{5}{*}{Erdos Renyi} & $0.1$ & $0.201$ & $0.550$ & $0.545$ & $\mathbf{0.571}$ & $0.555$ & $0.553$ & $0.564$ & $0.527$\\ 
& & $0.3$ & $0.619$ & $0.650$ & $0.659$ & $\mathbf{0.704}$ & $0.699$ & $0.678$ & $0.685$ & $0.571$\\ 
& & $0.5$ & $1.099$ & $0.750$ & $0.752$ & $\mathbf{0.806}$ & $0.791$ & $0.781$ & $0.781$ & $0.587$\\ 
& & $0.7$ & $1.735$ & $0.850$ & $0.851$ & $\mathbf{0.897}$ & $0.874$ & $0.876$ & $0.870$ & $0.611$\\ 
& & $0.9$ & $2.944$ & $0.950$ & $0.949$ & $\mathbf{0.967}$ & $0.948$ & $0.960$ & $0.955$ & $0.620$\\ 
\cmidrule{2-11}
& \multirow[c]{5}{*}{Power-law Graph} & $0.1$ & $0.201$ & $0.550$ & $0.545$ & $0.590$ & $\mathbf{0.609}$ & $0.580$ & $0.587$ & $0.586$\\ 
& & $0.3$ & $0.619$ & $0.650$ & $0.646$ & $\mathbf{0.715}$ & $0.713$ & $0.689$ & $0.700$ & $0.672$\\ 
& & $0.5$ & $1.099$ & $0.750$ & $0.745$ & $\mathbf{0.813}$ & $0.805$ & $0.801$ & $0.795$ & $0.701$\\ 
& & $0.7$ & $1.735$ & $0.850$ & $0.850$ & $\mathbf{0.890}$ & $0.883$ & $0.884$ & $0.876$ & $0.743$\\ 
& & $0.9$ & $2.944$ & $0.950$ & $0.949$ & $\mathbf{0.959}$ & $0.948$ & $0.939$ & $0.953$ & $0.753$\\ 
\cmidrule{2-11}
& \multirow[c]{5}{*}{Hierarchical} & $0.1$ & $0.201$ & $0.550$ & $0.540$ & $\mathbf{0.602}$ & $0.584$ & $0.577$ & $0.591$ & $0.534$\\ 
& & $0.3$ & $0.619$ & $0.650$ & $0.652$ & $\mathbf{0.730}$ & $0.711$ & $0.699$ & $0.724$ & $0.573$\\ 
& & $0.5$ & $1.099$ & $0.750$ & $0.753$ & $\mathbf{0.821}$ & $0.803$ & $0.808$ & $0.796$ & $0.612$\\ 
& & $0.7$ & $1.735$ & $0.850$ & $0.852$ & $\mathbf{0.893}$ & $0.883$ & $0.887$ & $0.881$ & $0.669$\\ 
& & $0.9$ & $2.944$ & $0.950$ & $0.949$ & $\mathbf{0.964}$ & $0.961$ & $0.960$ & $0.963$ & $0.756$\\ 
\midrule
\multirow[c]{20}{*}{\rotatebox[origin=c]{90}{\Large{Real-world}}}
& \multirow[c]{5}{*}{GrQc} & $0.1$ & $0.201$ & $0.550$ & $0.554$ & $\mathbf{0.577}$ & $0.576$ & $0.564$ & $0.561$ & N/A\tnote{\textasteriskcentered}\\ 
& & $0.3$ & $0.619$ & $0.650$ & $0.644$ & $\mathbf{0.720}$ & $0.718$ & $0.683$ & $0.685$ & N/A\\ 
& & $0.5$ & $1.099$ & $0.750$ & $0.744$ & $0.833$ & $\mathbf{0.834}$ & $0.798$ & $0.797$ & N/A\\ 
& & $0.7$ & $1.735$ & $0.850$ & $0.849$ & $\mathbf{0.908}$ & $0.892$ & $0.892$ & $0.878$ & N/A\\ 
& & $0.9$ & $2.944$ & $0.950$ & $0.951$ & $\mathbf{0.973}$ & $0.960$ & $0.971$ & $0.965$ & N/A\\ 
\cmidrule{2-11}
& \multirow[c]{5}{*}{HepTh} & $0.1$ & $0.201$ & $0.550$ & $0.553$ & $\mathbf{0.584}$ & $0.579$ & $0.568$ & $0.566$ & N/A\\ 
& & $0.3$ & $0.619$ & $0.650$ & $0.658$ & $\mathbf{0.733}$ & $0.716$ & $0.698$ & $0.703$ & N/A\\ 
& & $0.5$ & $1.099$ & $0.750$ & $0.751$ & $\mathbf{0.831}$ & $0.825$ & $0.793$ & $0.793$ & N/A\\ 
& & $0.7$ & $1.735$ & $0.850$ & $0.849$ & $\mathbf{0.917}$ & $0.893$ & $0.884$ & $0.876$ & N/A\\ 
& & $0.9$ & $2.944$ & $0.950$ & $0.950$ & $\mathbf{0.972}$ & $0.958$ & $0.962$ & $0.962$ & N/A\\ 
\cmidrule{2-11}
& \multirow[c]{5}{*}{Amazon Videos} & $0.1$ & $0.201$ & $0.550$ & $0.554$ & $0.598$ & $\mathbf{0.604}$ & $0.584$ & $0.585$ & N/A\\ 
& & $0.3$ & $0.619$ & $0.650$ & $0.649$ & $\mathbf{0.748}$ & $0.748$ & $0.706$ & $0.714$ & N/A\\ 
& & $0.5$ & $1.099$ & $0.750$ & $0.747$ & $\mathbf{0.854}$ & $0.836$ & $0.828$ & $0.812$ & N/A\\ 
& & $0.7$ & $1.735$ & $0.850$ & $0.849$ & $\mathbf{0.919}$ & $0.899$ & $0.906$ & $0.894$ & N/A\\ 
& & $0.9$ & $2.944$ & $0.950$ & $0.950$ & $\mathbf{0.975}$ & $0.962$ & $0.974$ & $0.967$ & N/A\\ 
\cmidrule{2-11}
& \multirow[c]{5}{*}{Amazon DVDs} & $0.1$ & $0.201$ & $0.550$ & $0.551$ & $0.576$ & $\mathbf{0.584}$ & $0.567$ & $0.570$ & N/A\\ 
& & $0.3$ & $0.619$ & $0.650$ & $0.652$ & $0.735$ & $\mathbf{0.739}$ & $0.699$ & $0.707$ & N/A\\ 
& & $0.5$ & $1.099$ & $0.750$ & $0.750$ & $\mathbf{0.830}$ & $0.829$ & $0.806$ & $0.802$ & N/A\\ 
& & $0.7$ & $1.735$ & $0.850$ & $0.851$ & $\mathbf{0.906}$ & $0.897$ & $0.894$ & $0.886$ & N/A\\ 
& & $0.9$ & $2.944$ & $0.950$ & $0.952$ & $\mathbf{0.963}$ & $0.957$ & $0.957$ & $0.941$ & N/A\\ 
\bottomrule
\end{tabular}
\begin{tablenotes}
 \footnotesize
 \item[\textasteriskcentered] Due to the large size of real-world networks, performing experiments with \lappas~ was not viable.
\end{tablenotes}
\end{threeparttable}
\label{tab:rocs}
\end{table*}

\smallskip\noindent\textbf{Real-world Networks}:
We also test our method on real-world networks that have been extensively studied in the context of influence maximization before (Table~\ref{tab:stats}):
two co-authorship networks 
GrQc and HepTh~\cite{collaboration} and two co-purchase networks in Amazon, one containing videos and the other DVDs~\cite{amazon, aria-ties}. 

\subsection{Experiment Setting}\label{sec:setting}
We simulate $10$ cascades for each generated network.
For each simulated cascade, we choose a specific number of initially active nodes (seed nodes) at random. These numbers are selected in a way that makes it easier to produce cascades of reasonable size in the networks (with at least $1/4$ and at most $3/4$ of all nodes). For Hierarchical networks, $50$ and for the rest of synthetic networks $5$ seed nodes were chosen. For real-world networks, we selected $5\%$ of nodes at random since the networks' size greatly varied.
We then simulate a cascade using the \emph{triggering set} approach in Theorem 4.6 of ~\cite{kempe2003maximizing}: Each edge $(u,v)$ is kept with probability $w(u,v)$ and tossed otherwise. Then, any node reachable from any of the seed nodes is deemed active.

To produce differentially private perturbations for the final $x_v$ values, we choose the Randomized Response (RR) mechanism~\cite{rr-original}. Given a value $x \in \{0, 1\}$ an RR mechanism works as below:
\begin{equation}\label{eq:rr}
    \mathcal{RR}(x) = \begin{cases}
        x, &\text{with probability $\beta$};\\
        1, &\text{with probability $(1-\beta)/2$};\\
        0. &\text{with probability $(1-\beta)/2$}.
    \end{cases}
\end{equation}
where $\beta$ is a parameter that determines the rate by which respondents report the \emph{true} value $x$. A RR mechanism with parameter $\beta$ is a DP mechanism with $\varepsilon = \log{\left((1 + \beta)/(1 - \beta)\right)}$.

\begin{figure*}[!th]
     \centering
     \includegraphics[width=0.99\textwidth]{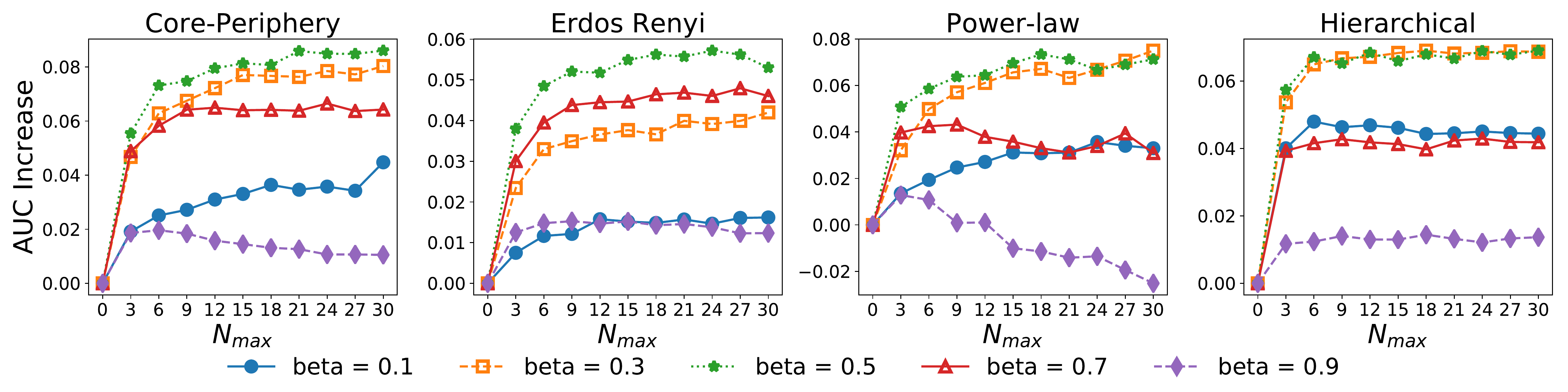}
     \caption{Impact of maximum allowed nodes ($\nmax$) in each DAG on increase in AUC scores relative to $\nmax = 3$.}
     \label{fig:node_roc}
\end{figure*}
\begin{figure*}[!tbh]
	\centering
	\includegraphics[width=0.99\textwidth]{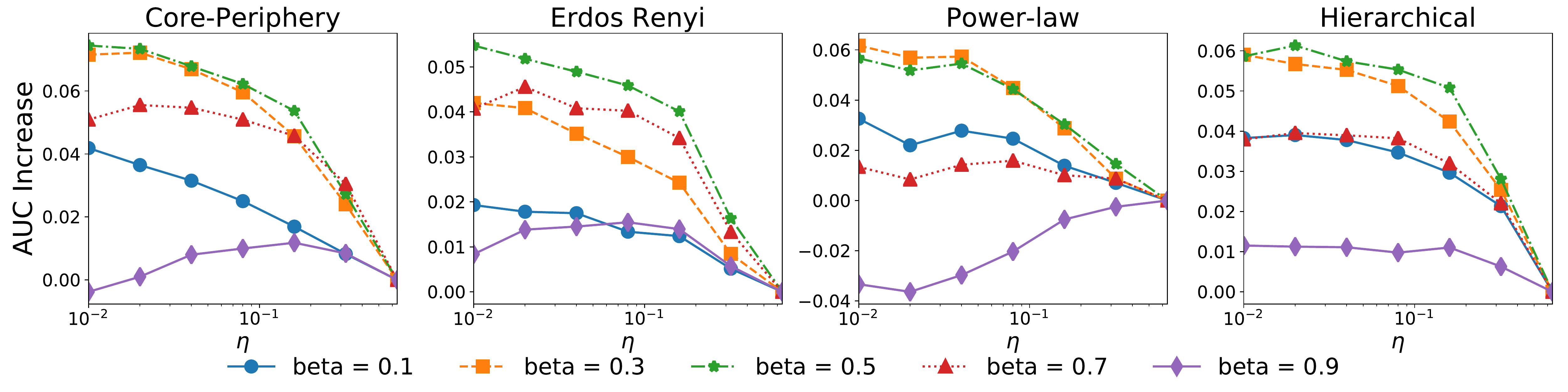}
	\caption{Impact of threshold on $\In$, $\eta$ on increase in AUC scores relative to $\eta = 0.01$.}
	\label{fig:prob_roc}
\end{figure*}
\def\expacc{\mathbb{E}[\text{Acc}]}
\begin{figure*}[!tbh]
    \centering
    \includegraphics[width=1\textwidth]{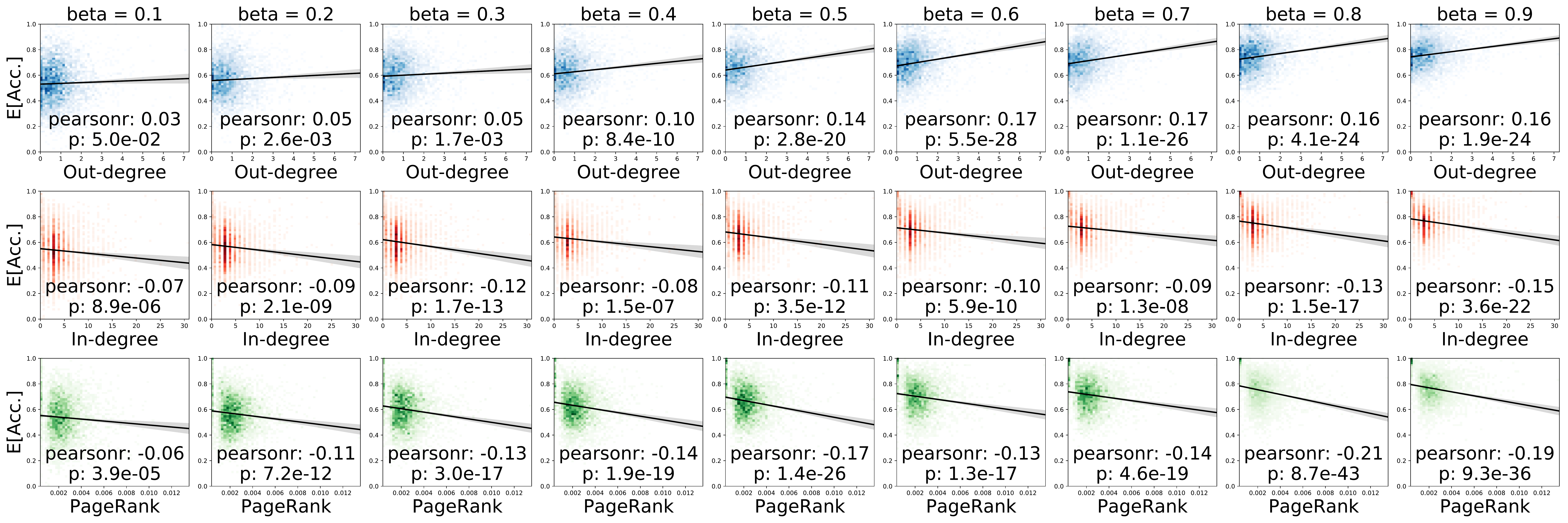}
    \caption{The correlation between $\expacc$ and from top to bottom (1) out-degree, (2) in-degree, and (3) PageRank.}
    \label{fig:metrics_vs_acc}
\end{figure*}
\begin{figure*}[!tbh]
    \centering
    \includegraphics[width=1\textwidth]{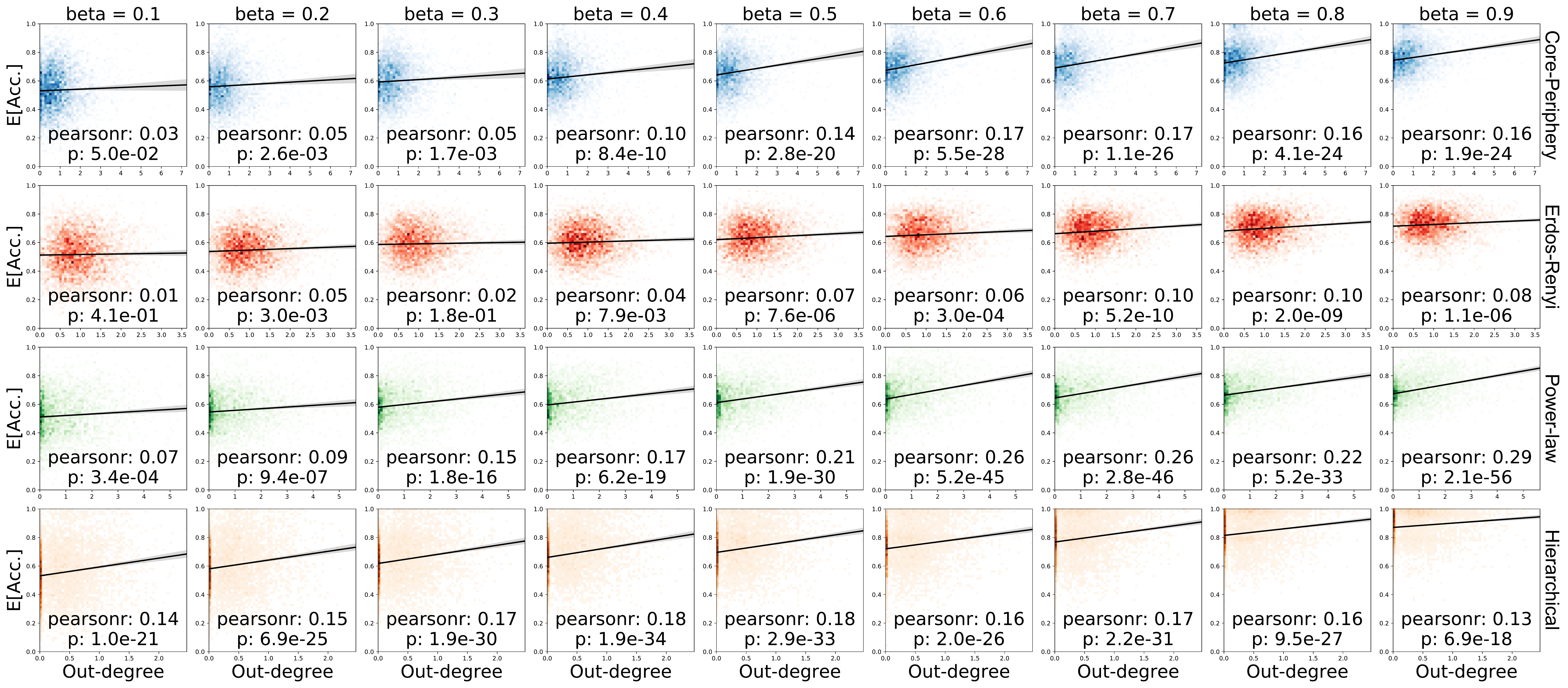}
    \caption{The correlation between out-degree ($x$-axis) and $\expacc$. Each column represents experiments with the same $\beta$.
    }
    \label{fig:graphs_vs_acc}
\end{figure*}
To do the inference, we need to extract DAGs for each node. For random DAGs (used in \ornd~and \cornd), we use the average number of nodes in their corresponding greedy DAGs (used in \odag~and \codag) as node capacity. To optimize \eqref{eq:obj_constraint}, we use ALGLIB\footnote{ALGLIB (www.alglib.net), Sergey Bochkanov} library.
\subsection{Going Above the Bound}
In Section~\ref{sec:problem} we proved that, without extra information about the contagion, it is impossible to achieve an AUC higher than $1 - (1 - \delta)/(1 + e^\varepsilon)$ on binary attributes perturbed by a $(\varepsilon,\delta)$-DP mechanism. Here, we test this guarantee after the contagion graph is known to an adversary. For $5$ values of RR's $\beta$
we perform inference using methods described in Section~\ref{sec:baseline} and report the resulting AUC scores, the corresponding $\varepsilon$ and theoretical upper bounds in Table~\ref{tab:rocs}. 

As expected, we observe that the AUC scores for \bayes is close to the theoretical upper bound. Furthermore, note that in almost all cases \codag~achieves the highest AUC score and in\emph{all} cases the value is beyond the AUC upper bound with DP guarantees. For the medium range of $\beta$
this difference becomes significant. We also observe that in cases where $\beta$ is extremely low and the perturbed data is exremely noisy, \odag~ tends to outperform \codag. We believe that this is due to the deteriorating quality of the bound on $\sum x_v/n$ when the added noise becomes too big.
It is also worth noting that our method's capability to go beyond the bound does not come trivially, since \lappas~ fails to do so in all but one cases, even as it receives \emph{more} information as input (number of seeds, $k$). Finally, notice that in virtually all cases, dropping the greedy DAGs or the constraint on $\sum x_v/n$ hurts performance.
\subsection{Impacts of DAG's size on Inference Quality}
We now test the impact of changing $\eta$ (the threshold on nodes' influence) or $\nmax$ (maximum allowed nodes) of DAGs on the quality of inferred values. 

We iterate over an increasing $\nmax$ with a fixed RR's $\beta$ across $4$ synthetic networks. Since we are interested in the change in resulting AUC scores and not their absolute values, we move each resulting curve so that their starting points would land on $0.00$. 
Similarly for $\eta$, we start from $\eta = 0.01$ and exponentially grow the threshold to $0.64$ while $\beta$ is fixed. To make it easier to observe the change in AUC scores, we collocate the ending point of all curves on $0.00$.

The results for $\nmax$ and $\eta$ are depicted in Figures~\ref{fig:node_roc} and \ref{fig:prob_roc} respectively. Intuitively, one might expect that larger DAGs will always lead to more accurate results, albeit with some additional computational cost. This is \emph{almost} true in both experiments. Interestingly, in some of the networks
when the perturbation is minimal ($\beta = 0.9, \varepsilon=2.94$), the quality actually drops as DAGs grow larger. This means that having a more \emph{selective} DAG, with more restrictive criteria, can sometimes yield better results.

\subsection{Who is More Vulnerable?}
The AUC value describes the effectiveness of our inference algorithm but does not provide details on individual vertex level.
Here we look into attributes that can indicate a node's vulnerability to such an inference attack.
Without access to ground truth values of $\x_v$, finding the best cutoff threshold of $x_v$ for classification is not possible. As an alternative, we calculate the \emph{expected accuracy} of a single node's inferred value if the threshold $\theta$ is selected at random from $[0,1]$:
\begin{equation}\label{eq:exp_acc}
    \expacc = \mathrm{E}_{\theta \sim [0, 1]}\left[ I( \x_v = 1)x_v + I(\x_v = 0)(1 - x_v) \right],
\end{equation}
where $I(\cdot)$ is an indicator function.
Note that we are not interested in the absolute values of $\expacc$, but the relative difference between different nodes, in order to compare nodes' vulnerability to our attack.

We tested $3$ node attributes related to the centrality and importance of a node in a complex network:
(1) (Weighted) Out-degree, (2) (Weighted) In-degree, and (3) PageRank~\cite{pagerank}.
We limit our visualizations to Core-Periphery networks due to limited space. Results are similar for the other 3 networks. 
We test $9$ values of $\beta = \{ 0.1, 0.2, \cdots, 0.9 \}$ and infer the real values using \codag. In Figure~\ref{fig:metrics_vs_acc} the resulting $\expacc$ for each node $v$ is plotted against the $3$ metrics mentioned above along with correlation analysis printed on each plot. Each column contains plots with a similar $\beta$ value, printed on the top.
Among the $3$ metrics, two of them (in-degree and PageRank) show a negative correlation, while a strong correlation is observed between $\expacc$ and out-degree which grows even stronger as $\beta$ is increased and the added noise becomes minimal.
A reasonable explanation is that when $v$ has a high impact on many nodes in the network, it essentially sends signals of its true value to those nodes during the spread of the contagion. These signals, each insignificant on its own, can collectively reveal the true value of $x_v$ to a great extent. Of course, these signals only grow stronger when the amount of added noise to each report is reduced as $\beta$ grows bigger. In contrast, there might be many signals on the influence received by a node with a high in-degree, but without that influence spreading around, we only have the report by that node as a single indicator of its true value which is masked well by conventional privacy protection techniques.

We repeat this experiment on other synthetic networks to see if we observe similar results.
In Figure~\ref{fig:graphs_vs_acc}, $\expacc$ is plotted against out-degree of all nodes for $\beta$ between $0.1$ and $0.9$. A similar trend is visible across all $4$ types of network, with different intensity. 
In the $3$ networks where the distribution of influence is skewed,
namely Core-Periphery, Power-law and Hierarchical, the correlation is strong.
Moreover, as $\beta$ increases and the added noise become smaller,
the correlation grows stronger. In Erdos-Renyi networks, where degrees are distributed more evenly, the difference in $\expacc$ among nodes is minor and as a result, the correlation is weaker.
Notice that in Hierarchical network, this \emph{unequal} distribution of influence reaches its extreme point, where the majority of the nodes are leaves, with \emph{no} impact on others, while a handful of nodes at the top  influence many others. This radical difference
leads to a constant unbalance between $\expacc$ of the two groups across all values of $\beta$.


%% file: 60_conclusion.tex
\section{Conclusion and Future Work}\label{sec:conclusion}
In this work, we provide further evidence that a privacy-protecting measure that is oblivious to socially contagious properties of attributes is unlikely to provide guarantees in practice as advertised.
There are two obvious directions for future research: 1) design a privacy protection mechanism that is aware of the contagious property of some attributes and employs models of contagions in its design; Our method can be used to test if that method succeeds. 2) extend the results to other contagion models. Notice that the extension to any progressive contagion model (where active nodes stay active) with a differentiable formula (e.g, Independent Cascade model) is easily possible. However, ideas specific to those models are required for efficient implementation and model design.

